\documentclass[submission,copyright,creativecommons]{eptcs}
\providecommand{\event}{QPL 2019} 
\usepackage{amsmath,amssymb,amsthm,braket,tikz,graphicx}

\newtheorem{theorem}{Theorem}
\newtheorem{proposition}[theorem]{Proposition}

\newtheorem{lemma}[theorem]{Lemma}

\newtheorem{definition}[theorem]{Definition}
\newtheorem{example}[theorem]{Example}
\newtheorem{remark}[theorem]{Remark}

\newcommand{\cat}[1]{\ensuremath{\mathbf{#1}}}
\newcommand{\id}[1][]{\ensuremath{\mathrm{id}_{#1}}}
\newcommand{\M}{\ensuremath{\mathbb{M}}}
\DeclareMathOperator{\CPs}{\ensuremath{\mathrm{CP}^*}}
\DeclareMathOperator{\Tr}{Tr}

\title{Ontological models for quantum theory as functors}
\def\titlerunning{Ontological models for quantum theory as functors}
\author{Alexandru Gheorghiu
  \institute{California Institute of Technology}
  \email{andrugh@caltech.edu}
 \and
  Chris Heunen
  \institute{University of Edinburgh}
  \email{chris.heunen@ed.ac.uk}}
\def\authorrunning{A. Gheorghiu \& C. Heunen}

\begin{document}
\sloppy

\maketitle
\begin{abstract}
  We interpret ontological models for finite-dimensional quantum theory as functors from the category of finite-dimensional Hilbert spaces and bounded linear maps to the category of measurable spaces and Markov kernels. 
  This uniformises several earlier results, that we analyse more closely: 
  Pusey, Barrett, and Rudolph's result rules out monoidal functors; 
  Leifer and Maroney's result rules out functors that preserve a duality between states and measurement; 
  Aaronson et al's result rules out functors that adhere to the Schr{\"o}dinger equation.
  We also prove that it is possible to have epistemic functors that take values in signed Markov kernels.
\end{abstract}

\section{Introduction}

Is the wavefunction of quantum theory an objective property of reality, or merely a statistical quantity associated with a probability distribution over the actual elements of reality? This question divides foundations of quantum theory into two camps -- the former theories are \emph{ontic}, whereas the latter theories are \emph{epistemic} -- and has occupied quantum foundations greatly. 
There are many ontic interpretations of quantum mechanics, such as many-worlds, de Broglie-Bohm, or modal theory.
But there are not many epistemic theories that fully reproduce the predictions of quantum mechanics. 
One example of an epistemic theory is Spekkens' toy model~\cite{spekkens2007evidence}, but that only considers a restricted version of quantum mechanics.
The difficulty in having an epistemic interpretation of quantum mechanics is partly explained by no-go theorems that constrain models attempting to reproduce quantum mechanics using classical probability distributions. Notable such obstructions are: Bell's theorem~\cite{bell1964einstein}, that rules out local models; the Kochen-Specker theorem~\cite{kochen1975problem}, that rules out noncontextual models; the Pusey-Barrett-Rudolph (PBR) theorem~\cite{puseybarrettrudolph:reality}, that rules out models in which independently prepared quantum states correspond to independent ontic states. 

Underlying all these investigations is the question: is it possible to have some translation from quantum theory to probability theory? Whether such a translation preserving certain structural aspects of quantum theory is possible explains whether quantum theory is ontic or epistemic. There is a branch of mathematics whose entire reason for being is to translate structure between different areas, namely category theory~\cite{leinster:categorytheory}. This suggests phrasing translation questions about possible ontological models as \emph{functors}, and that is exactly what this paper does. 

In Sections~\ref{sec:ontological} and~\ref{sec:operational} we recognise ontological models as functors from (the category of finite-dimensional Hilbert spaces and bounded linear maps modelling) finite-dimensional quantum theory to probability theory (as modelled by the category of Borel spaces and Markov kernels). The former category contains states and measurements as morphisms, and the latter category contains probability measures as morphisms, but both contain more morphisms, incorporating dynamics in a natural way.

We then ask the question whether such functors satisfying various properties can exist (see also~\cite{abramskyheunen:operational}). This language uniformises several earlier results, and lets us analyse their structure more closely.
\begin{itemize}
  \item Can there be an epistemic functor that preserves tensor products? Section~\ref{sec:monoidal} analyses the PBR theorem~\cite{puseybarrettrudolph:reality} in these terms to provide a negative answer.
  \item Can there be an epistemic functor that preserves duality between states and measurements? Section~\ref{sec:dagger} analyses~\cite{leifer2013maximally} categorically to rule out any such functor, never mind an epistemic one.
  \item Can there be an epistemic functor that preserves the Schr{\"o}dinger equation? Section~\ref{sec:equivariant} analyses~\cite{aaronson2013psi} categorically to show that there can be no maximally nontrivial such functor.
\end{itemize}
Moreover, the formulation in terms of functors naturally suggests other questions.
\begin{itemize}
  \item What if we change the target category? Section~\ref{sec:signed} shows that an epistemic model is possible when we move from Markov kernels to signed Markov kernels.
\end{itemize}
Epistemic functors also seem possible when using quantum measures rather than signed measures, but this runs into technical issues; see Appendix~\ref{sec:quantummeasure}.
We also leave open the following naturally suggested questions: Can there be a (co)limit-preserving epistemic functor? Can there be an (op)lax monoidal epistemic functor? Can there be an epistemic functor at all? Towards the latter question: there are epistemic models of quantum theory, e.g.\ those in~\cite{lewis2012distinct,aaronson2013psi}, but these mappings from quantum states to probability distributions are not functorial. As far as we are aware, no such mapping is known to exist.

\section{Ontological models} \label{sec:ontological}

Interpretations of quantum mechanics that describe an objective reality (realist interpretations) do so in the context of an \emph{ontological model}. Let us recall the standard definitions~\cite{leifer:review, aaronson2013psi}.

\begin{definition}
  An \emph{ontological model} is a Borel space\footnote{In~\cite{leifer:review}, an ontological model is defined simply as a measurable space, rather than a Borel space. However, a measurable space by itself does not have sufficient structure for certain properties of interest regarding ontological models. In particular, the notion of \emph{support} of a measure over the space needs to be defined, see Definition~\ref{def:maximallyepistemic} below. Thus we consider Borel spaces.} $\Lambda$, called the \emph{ontic space}. Write $\Sigma_\Lambda$ for its $\sigma$-algebra.
\end{definition}

\begin{definition}
  An \emph{ontological theory} of quantum mechanics is a theory satisfying the following:
  \begin{enumerate}
    \item Each finite-dimensional Hilbert space $H$ has an associated ontological model $(\Lambda, \Sigma_{\Lambda})$;
    \item Each state $\ket{\psi} \in H$ has an associated probability measure $\mu_{\psi} \colon \Sigma_{\Lambda} \rightarrow [0,1]$ with $\mu_{\psi}(\Lambda) = 1$;
    \item Each orthonormal measurement $M = \{ \ket{\phi_1}, \Ket{\phi_2}, \ldots \ket{\phi_{\dim(H)}} \}$ has a set of response functions $\{ \xi_{k,M} \colon \Lambda \rightarrow [0,1] \mid 1 \leq k \leq \dim(H)\}$ satisfying:
	\begin{align*}
	  \forall \ket\psi \in H \colon & \int_{\Lambda} \xi_{k,M}(\lambda) d \mu_{\psi}(\lambda) = | \braket{\phi_k \mid \psi}|^2 \\
	  \forall \lambda \in \Lambda \colon & \sum_{i = 1}^d \xi_{k,M} (\lambda) = 1 
    \end{align*}
  \end{enumerate}
\end{definition}

\begin{example} \label{ex:universal}
  The simplest example of an ontological theory is the following:
  \begin{enumerate}
    \item $\Lambda = \mathbb{CP}^{\dim(H)-1}$ is the complex projective space of $H$ under its Borel $\sigma$-algebra;
	\item $\mu_{\psi}(U) = \chi_{U}(\psi)$, writing $\chi_{U}$ for the indicator function of the subset $U \subseteq \Lambda$;
	\item $\xi_{k,M}(\lambda) = | \braket{\phi_k \mid \lambda}|^2$.
  \end{enumerate}
  Of course, this is merely a restatement of the original Hilbert space formulation.
\end{example}

Ontological theories (also called \emph{interpretations}) of quantum mechanics come in two types: ontic and epistemic. The wavefunction is regarded as an objective property of reality in the former, and as a statistical quantity in the latter.

\begin{definition}
  An ontological theory of quantum mechanics is \emph{epistemic} if there exist states $\ket{\psi}, \ket{\phi} \in H$, satisfying $0 < | \braket{\psi\mid\phi}| < 1$ and
  $
	D(\mu_{\psi}, \mu_{\phi}) < 1
  $,
  where
  \[
	D(\mu_{\psi}, \mu_{\phi}) = \sup\limits_{\Omega \in \Sigma_{\Lambda}} | \mu_{\psi}(\Omega) - \mu_{\phi}(\Omega) |
  \]
  is the \emph{variational distance} between the probability measures $\mu_{\psi}$ and $\mu_{\phi}$ associated to $\ket{\psi}$ and $\ket{\phi}$.

  An ontological theory is \emph{ontic} when it is not epistemic.
\end{definition}

In essence, the previous definition says an ontological theory is epistemic when there exists a pair of distinct but overlapping states, whose associated distributions over the ontic space have overlapping support. 
Example~\ref{ex:universal} is ontic.

The motivation for considering epistemic theories is explained in detail in~\cite{leifer:review}. Briefly, epistemic theories attempt to address the following question: to what extent can quantum uncertainty be explained as lack of knowledge of fundamental physical degrees of freedom? The Bayesian view of probability is that it represents a state of knowledge about a system or a process. In the case of ontological theories, the distribution associated to a quantum state encodes the uncertainty in the underlying ontic state. This is analogous to statistical mechanics, where a macroscopic property like the temperature or pressure of a gas corresponds to a probability distribution over the system's phase space (representing the space of possible position and momenta of the gas particles). 
The point of epistemic theories is to then argue that quantum states that cannot be perfectly distinguished (\textit{i.e.} states with nonzero overlap) should correspond to overlapping probability distributions. That is to say that the uncertainty in discerning which quantum state characterises a system stems from there being ontic states compatible with multiple quantum states.

The condition for an ontological theory to be epistemic does not specify which pair of states should have overlapping distributions; the requirement is merely that such a pair exist. It may be more natural to require that the overlap between states be completely explained by the overlap in their associated distributions; this is called \emph{maximally epistemic}~\cite{leifer2013maximally}. 
Similarly, it may be more natural to require that whether states overlap at all is completely explained by whether their associated distributions overlap at all; this is called \emph{maximally nontrivial}~\cite{aaronson2013psi}.

\begin{definition}\label{def:maximallyepistemic}
  For a state $\psi \in H$ with ontic space $\Lambda$, let $\Lambda_\psi = \{ \lambda \in \Lambda \mid \lambda \in U \in \Sigma_{\Lambda} \implies \mu_{\psi}(U) > 0 \}$ be the support of $\mu_\psi$.
  An ontological theory of quantum mechanics is maximally epistemic if for all $\psi, \phi \in H$:
  \begin{equation}\label{eq:maximallyepistemic}
    \mu_\psi(\Lambda_\phi)
    = | \braket{\phi \mid \psi}|^2
  \end{equation}
  It is \emph{maximally nontrivial} when $\braket{\phi \mid \psi}=0$ if and only if
  $\mu_\psi(\Lambda_\phi)=0$.
\end{definition}

\section{Operational models} \label{sec:operational}

We think of a category as consisting of (models of) physical systems and processes that can be composed in sequence and in parallel~\cite{coeckekissinger,heunenvicary}.
We first axiomatise probabilistic measurements in such a setting.

\begin{definition}\label{def:operationalcategory}
  An \emph{operational category} consists of:
  \begin{itemize}
	\item a monoidal category $\cat{C}$, with tensor product $\otimes$ and unit $I$;
	\item an object $2$ in $\cat{C}$, called the \emph{distinguishing object};
	\item a set $\Omega$, whose elements are called \emph{probabilities};
	\item a function $\langle - \rangle \colon \cat{C}(I,2) \to \Omega$ called \emph{evaluation}.
  \end{itemize}
\end{definition}

Maps $X \to 2$ are also called \emph{measurements}, maps $I \to X$ \emph{states}, and maps $I \to 2$ \emph{abstract probabilities}. 
We will often fix $\Omega$ to be the unit interval $[0,1]$, in which case we also speak of a \emph{concrete operational category}, to justify the name `probabilities' for elements of $\Omega$.
One might assume much more structure than the above definition.
For example, the set of probabilities $\Omega$, might be taken to be a partially ordered set, a monoid, or even a semiring.
Similarly, the distinguishing object $2$ might be assumed to be a generator, or a coproduct $I+I$~\cite{jacobsetal:effectus}.
Finally, the category $\cat{C}$ might be assumed to be compact, or dagger~\cite{coeckekissinger,heunenvicary}.
Here we will only assume the bare minimum of the above definition.

The prototypical example of an operational category is standard quantum theory~\cite{coeckeheunenkissinger:cp}.

\begin{example}\label{ex:quantum}
  The category $\cat{FHilb}$ of finite-dimensional Hilbert spaces and (bounded) linear maps is operational with  distinguishing object $\mathbb{C}^2$, abstract probabilities $[0,1]$, and the Born rule 
  \[
    \langle \psi \rangle = |a|^2 \quad\text{if}\quad \psi(1)=(a,b)
  \]
  as evaluation function for $\psi \colon \mathbb{C} \to \mathbb{C}^2$.
  States correspond to Hilbert space vectors, and measurements are projective measurements with 2 outcomes.
  All negative results below apply equally well to the category of all Hilbert spaces and bounded linear maps, regardless of dimension.
\end{example}

We will be interested in other operational settings, such as probability theory. Here, the model of a physical system is its set of (ontic) states, and a physical processes simply evolves ontic states.

\begin{example}\label{ex:srel}
  A \emph{Markov kernel} from a measurable space $(X,\Sigma_X)$ to a measurable space $(Y,\Sigma_Y)$ is a probability-measure-valued function $f \colon X \times \Sigma_Y \to [0,1]$ such that $f(-,V) \colon X \to [0,1]$ is a bounded measurable function for each $V \in \Sigma_Y$, and $f(x,-) \colon \Sigma_Y \to [0,1]$ is a probability measure for each $x \in X$.
  Measurable spaces and Markov kernels form a category $\cat{SRel}$ with composition $(g \circ f)(x,W) = \int g(y,W) f(x,\mathrm{d}y)$, and Dirac measures $\id[X](x,U)=1$ for $X \ni x \in U \in \Sigma_X$ and $\id[X](x,U)=0$ for $X \ni x \not\in U \in \Sigma_X$ as identities~\cite{panangaden:labelledmarkovprocesses,panangaden:srel}.
  Here, the notation $f(x,\mathrm{d}y)$ is short for $\mathrm{d}f(x,-)$.

  The category $\cat{SRel}$ is (symmetric) monoidal. In the abstract, because it is the Kleisli category of the monoidal probability Giry monad~\cite{jacobs:measurableeffects}. We describe the monoidal structure concretely.
  The tensor product $(X,\Sigma_X) \otimes (Y,\Sigma_Y)$ of objects is carried by $X \times Y$ and furnished with the $\sigma$-algebra $\Sigma_{X \times Y}$ generated by the sets $U \times V$ for $U \in \Sigma_X$ and $V \in \Sigma_Y$.
  The tensor unit is the singleton set $I=\{*\}$ with its unique $\sigma$-algebra.
  The tensor product $f \otimes f' \colon X \otimes X' \to Y \otimes Y'$ of Markov kernels $f \colon X \to Y$ and $f' \colon X' \to Y'$ is determined by $((x,x'),V \times V') \mapsto f(x,V) \cdot f'(x',V')$.

  States $\psi \colon I \to X$ in $\cat{SRel}$ correspond to probability measures $\Sigma_X \to [0,1]$ on $X$.
  As distinguishing object we take $2=\{0,1\}$, with the discrete $\sigma$-algebra $\Sigma_2=\{\emptyset,\{0\},\{1\},2\}$.
  Measurements $X \to 2$ correspond to Markov kernels $f \colon X \times \Sigma_2 \to [0,1]$, which are completely determined by a measurable function $x \mapsto f(x,0) \colon X \to [0,1]$.
  Probabilities $f \colon I \to 2$ thus correspond exactly with elements $f(*,0)$ of $[0,1]$. Thus the category $\cat{SRel}$ becomes operational under $\Omega=[0,1]$ with evaluation $\langle f \rangle = f(*,0)$.

  Write \cat{BoRel} for the full subcategory of Borel spaces. It inherits all structure of \cat{SRel} described above.
\end{example}

As mentioned when we defined ontological models, we will be interested in ontic spaces that can be represented as Borel spaces.
For this reason, in examining realist interpretations of quantum mechanics from this categorical perspective, we will consider the subcategory \cat{BoRel} of \cat{SRel} to correspond to ontological models. 
Since $\cat{FHilb}$ corresponds to quantum mechanics, an interpretation (or an ontological theory) will correspond to some sort of translation from $\cat{FHilb}$ to $\cat{BoRel}$. This translation should preserve the empirical predictions of the Born rule. The most natural translation from the categorical perspective is one that preserves the categorical structure of composition: a functor. 
Some formulations of ontological models also assume that unitary evolution on the quantum side, is mapped to a stochastic evolution on the ontological side. In fact, as noted in~\cite{leifer:review}, the evolution of ontic states can be modelled through a Markov kernel. Additionally imposing that composing unitaries is preserved at the ontological level, recovers exactly the functorial map in question. 

\begin{definition}\label{def:operationalmodel}
  Let $\cat{C}$ and $\cat{D}$ be operational categories with the same probabilities $\Omega$. 
  An \emph{operational model} is a functor $F \colon \cat{C} \to \cat{D}$ that satisfies $F(I_{\cat{C}})=I_{\cat{D}}$, $F(2_{\cat{C}})=2_{\cat{D}}$ and $\langle F(\psi) \rangle_{\cat{D}} = \langle \psi \rangle_{\cat{C}}$.
  \footnote{We follow the categorical/logical/model-theoretic convention that terms the domain a ``theory'', and the functor a ``model'' or ``interpretation'', rather than the physical convention where a ``theory'' is a class of ``models''.}
\end{definition}

We are specifically interested in $\cat{C} = \cat{FHilb}$ and $\cat{D} = \cat{BoRel}$. In this case, the only fundamental difference between having an ontological theory and having an operational model is that the operational model forces the translation from $\cat{FHilb}$ to $\cat{BoRel}$ to preserve composition. 

To define when operational models are ontic or epistemic, we first need to say in terms of operational categories when probability measures \emph{do not} overlap. This leads to \emph{anti-distinguishability}\footnote{In the literature, anti-distinguishability is also referred to as \emph{state discrimination}~\cite{barnett2009quantum}.}.

\begin{definition}\label{def:distinguishability}
  Let $\cat{C}$	be a concrete operational category, and $\Psi \subseteq \cat{C}(I,A)$ a collection of states. A measurement $\chi \colon A \to 2$ \emph{anti-distinguishes} a fixed state $\psi \in \Psi$ if
  \[
    \langle \chi \circ \psi \rangle = 0\text,
    \qquad
    \sum_{\phi \neq \psi \in \Psi} \langle \chi \circ \phi \rangle = 1\text.
  \]
  A state $\psi \in \Psi$ is \emph{anti-distinguishable} within $\Psi$ if there is a measurement that anti-distinguishes it.
  Finally, $\Psi$ is \emph{anti-distinguishable} if each $\psi \in \Psi$ is anti-distinguishable.
\end{definition}

Analogously, the more familiar concept of ``distinguishability'' is defined in the same way, but having $0$ and $1$ swapped, so that $\langle \chi \circ \psi \rangle = 1$ and $\sum_{\phi \neq \psi \in \Psi} \langle \chi \circ \phi \rangle = 0$. 
Anti-distinguishability will be of more interest to us.

Probability measures correspond to states in $\cat{BoRel}$.
According to the above definition, two probability measures $\psi$ and $\phi$ are anti-distinguishable precisely when there exists a measurement $\chi$ that satisfies $0=\langle \chi \circ \psi \rangle=\int \chi(*,0) \mathrm{d}\psi$ and $1=\langle \chi \circ \phi \rangle=\int \chi(*,0) \mathrm{d}\phi$.
Because $\int_\Lambda \mathrm{d}\mu = \int_{\Lambda_\mu} \mathrm{d}\mu$ by definition of support, this means that $\chi$ assigns measure $0$ to the support of $\psi$ almost everywhere with respect to $\phi$, and assigns measure $1$ to the support of $\phi$ almost everywhere with respect to $\psi$. In other words, $\psi$ and $\phi$ are exactly non-overlapping measures.
Operationally, the measurement $\chi$ can be thought of as an experiment that samples from a given distribution and always rejects $\psi$, but always accepts $\phi$.
Notice that if there are only two distributions, distinguishability and anti-distinguishability are equivalent.

\begin{definition}\label{def:ontic}
  An operational model $F$ is \emph{ontic} when it maps distinct states $\psi\neq \phi$ in $\cat{C}$ to (anti-) distinguishable states $F(\psi),F(\phi)$ in \cat{BoRel}; otherwise it is \emph{epistemic}.
\end{definition}

\section{Monoidal operational models}\label{sec:monoidal}

This section is the first of several considering whether operational models with certain extra properties can exist. The property under scrutiny in this section is preserving tensor products, that is, we set out to establish a categorical version of the PBR theorem (see Appendix~\ref{sec:pbr} for a brief discussion of the original PBR theorem). 
Recall that a functor $F \colon \cat{C} \to \cat{D}$ is \emph{monoidal} when there are a natural isomorphisms $F_{A,B} \colon F(A) \otimes F(B) \to F(A \otimes B)$ and a morphism $F_0 \colon I \to F(I)$ satisfying certain coherence requirements. 
This means that if $\psi \colon I \to A$ is a state in $\cat{C}$, then $F(\psi) \circ F_0 \colon I \to F(A)$ is a state in $\cat{D}$; by abuse of notation we will simply write $F(\psi)$ for this state.
We will refer to an operational model in which the functor is monoidal as a \emph{monoidal operational model}.

Let us start with some properties of anti-distinguishability that hold in any monoidal operational model. We will then establish some properties of anti-distinguishability specific to $\cat{BoRel}$. 

\begin{lemma}\label{lem:preservedistinguishability}
  Let $F \colon \cat{C} \to \cat{D}$ be a monoidal operational model with $0\neq1 \in \Omega$, and $\Psi \subseteq \cat{C}(I,A)$ be a collection of states. If a measurement $\chi \colon A \to 2$ anti-distinghuishes $\psi \in \Psi$, then $F(\chi) \colon F(A) \to 2$ anti-distinguishes $F(\Psi) = \{F(\phi) \mid \phi \in \Psi\}$. 
  Therefore, if $\psi$ is anti-distinguishable in $\Psi$, then $F(\psi)$ is anti-distinguishable in $F(\Psi)$; and if $\Psi$ is anti-distinguishable, then so is $F(\Psi)$.
\end{lemma}
\begin{proof}
  Follows directly from Definitions~\ref{def:operationalmodel} and~\ref{def:distinguishability}.
\end{proof}

If $F \colon \cat{C} \to \cat{D}$ is an operational model, and $\psi \colon I \to A$ a state in $\cat{C}$, then there is a state $(F(\psi) \otimes F(\psi)) \circ \lambda \circ F_0 \colon I \to F(A) \otimes F(A)$; where $\lambda$ is the left unitor of $\cat{D}$. We will supress the coherence isomorphisms, which the following lemma justifies, and simply write $F(\psi)^{\otimes 2}$, and inductively define $F(\psi)^{\otimes n}$ similarly. Similarly, if $\chi \colon I \to 2$ is a measurement in $\cat{C}$, write $F(\chi)$ for the induced measurement $I \to 2$ in $\cat{D}$. 

\begin{lemma}\label{lem:preserventensordistinguishable}
  Let $F \colon \cat{C} \to \cat{D}$ be a monoidal operational model with $\Omega=[0,1]$, and $\Psi \subseteq \cat{C}(I,A)$ a collection of states.
  If $\{\psi^{\otimes n} \mid \psi \in \Psi\}$ is anti-distinguishable, then so is $\{F(\psi)^{\otimes n} \mid \psi \in \Psi\}$.
\end{lemma}
\begin{proof}
  Fix $\psi \in \Psi$, and say $\chi \colon A^{\otimes n} \to 2$ satisfies $\langle \chi \circ \psi^{\otimes n}\rangle =0$ and $\sum_{\psi \neq \phi \in \Psi} \langle \chi \circ \phi^{\otimes n}\rangle=1$. Now, we have to be slightly more precise about tensor products of states under $F$. Because $F$ is a monoidal functor, the following diagram commutes.
  \[\begin{tikzpicture}[xscale=4,yscale=1.5]
  	\node (tl) at (0,2) {$I$};
  	\node (t) at (1,2) {$I^{\otimes n}$};
  	\node (l) at (0,1) {$F(I)$};
  	\node (m) at (1,1) {$F(I)^{\otimes n}$};
  	\node (r) at (2,1) {$F(A)^{\otimes n}$};
  	\node (b) at (1,0) {$F(I^{\otimes n})$};
  	\node (br) at (2,0) {$F(A^{\otimes n})$};
  	\draw[->] (tl) to node[above]{$\lambda_I^n$} (t);
  	\draw[->] (tl) to node[left]{$F_0$} (l);
  	\draw[->] (l) to node[below]{$\lambda_{F(I)}^n$} (m);
  	\draw[->] (t) to node[right]{$F_0^{\otimes n}$} (m);
  	\draw[->] (m) to node[below]{$F(\psi)^{\otimes n}$} (r);
  	\draw[->] (m) to node[right]{$F_n$} (b);
  	\draw[->] (b) to node[below]{$F(\psi^{\otimes n})$} (br);
  	\draw[->] (r) to node[right]{$F_n$} (br);
  	\draw[->] (l) to[out=-90,in=180] node[left]{$F(\lambda^n)$} (b);
  \end{tikzpicture}\]
  Hence 
  \begin{align*}
    \langle F(\chi) \circ F(\psi^{\otimes n} \circ \lambda^n) \circ F_0 \rangle 
    = \langle F(\chi \circ \psi^{\otimes n} \circ \lambda^n) \circ F_0 \rangle 
    = \langle \chi \circ \psi^{\otimes n} \rangle = 0
  \end{align*}
  and similarly $\sum_{\phi \neq \psi \in \Psi} \langle F(\chi) \circ F(\phi) \rangle=1$.
\end{proof}

\noindent The next two lemmas concern specific properties of $\cat{BoRel}$.

\begin{lemma}\label{lem:twotensordistinguishablesrel}
  Let $\phi,\psi \in \cat{BoRel}(I,A)$ be states in $\cat{BoRel}$.	
  If $\{\phi \otimes \phi, \phi \otimes \psi, \psi \otimes \phi, \psi \otimes \psi\}$ is anti-distinguishable, then so is $\{\phi,\psi\}$.
\end{lemma}
\begin{proof}
  Say $\chi \colon A \otimes A \to 2$ satisfies
  $\langle \chi \circ (\psi \otimes \psi) \rangle = 0$ and
  \[
    \langle \chi \circ (\phi \otimes \phi) \rangle +
    \langle \chi \circ (\phi \otimes \psi) \rangle +
    \langle \chi \circ (\psi \otimes \phi) \rangle = 1\text.
  \]
  By Example~\ref{ex:srel} 
  $\langle \chi \circ (\psi \otimes \psi) \rangle 
    = \int_{A^2} \chi(a_1,a_2) \mathrm{d}\psi(a_1) \mathrm{d}\psi(a_2)$.
  Because $\chi$ and $\psi$ are nonnegative, the first equation therefore implies 
  \[
	\int_A \chi(a_1,a_2) \mathrm{d}\psi(a_1) 
	= 0 =
	\int_A \chi(a_1,a_2) \mathrm{d}\psi(a_2) 
  \]
  for all $a_1,a_2 \in A$.
  Thus $\psi$ vanishes almost everywhere and $\langle \chi \circ (\phi \otimes \psi) \rangle = 0 = \langle \chi \circ (\psi \otimes \phi)\rangle$. The second equation similarly implies
  \[
    1 
    = \langle \chi \circ (\phi \otimes \phi) \rangle +
       \langle \chi \circ (\phi \otimes \psi) \rangle +
       \langle \chi \circ (\psi \otimes \phi) \rangle 
     = \int_{A^2} \chi(a_1,a_2) \mathrm{d}\phi(a_1) \mathrm{d}\phi(a_2)\text, 
  \]
  so that $\langle \chi \circ (\psi \otimes \psi) \rangle =0$ and $\langle \chi \circ (\phi \otimes \phi) \rangle = 1$.

  Now define $\chi' \colon A \to 2$ by
  $\chi'(a) = \int_A \chi(a,a_2) \mathrm{d}\psi(a_2)$.
  Then 
  $\langle \chi' \circ \psi \rangle=0$ and $\langle \chi' \circ \phi \rangle=1$ by construction, so $\chi'$ anti-distinguishes $\phi$ and $\psi$.
\end{proof}

\begin{lemma}\label{lem:ntensordistinguishablesrel}
  Let $\phi,\psi \in \cat{BoRel}(I,A)$ be states in $\cat{BoRel}$, and $n>0$ a natural number.
  If $\{\phi^{\otimes n}, \psi^{\otimes n}\}$ is anti-distinguishable, then so is $\{\phi,\psi\}$.
\end{lemma}
\begin{proof}
  Say $\chi \colon A^{\otimes n} \to 2$ satisfies
  $\langle \chi \circ \phi^{\otimes n} \rangle=0$ and
  $\langle \chi \circ \psi^{\otimes n} \rangle=1$.
  By Example~\ref{ex:srel}:
  \begin{align*}
    \langle \chi \circ \phi^{\otimes n} \rangle 
    &= \int_{A^n} \chi(a_1,\ldots,a_n) \mathrm{d}\phi_1(a_1) \cdots \mathrm{d}\phi_n(a_n) = 0\text, \\
    \langle \chi \circ \psi^{\otimes n} \rangle 
    &= \int_{A^n} \chi(a_1,\ldots,a_n) \mathrm{d}\psi_1(a_1) \cdots \mathrm{d}\psi_n(a_n) = 1\text.
  \end{align*}
  Because $\chi$ and the measures $\psi_i$ and $\phi_i$ are positive, it follows that for any $a_2,\ldots,a_n \in A$:
  \[
    \int_A \chi(a,a_2,\ldots,a_n) \mathrm{d}\phi(a_1,\ldots,a_n)=0\text.
  \]
  Define $\chi' \colon A \to 2$ by
  $
    \chi'(a) = \int_{A^{n-1}} \chi(a,a_2,\ldots,a_n) \mathrm{d}\psi_2(a_2) \cdots \mathrm{d}\psi_n(a_n)
  $.
  Then, clearly, $\langle \chi \circ \phi \rangle = \langle \chi \circ \phi^{\otimes n} \rangle = 0$, and $\langle \chi \circ \psi \rangle = \langle \chi \circ \psi^{\otimes n} \rangle = 1$.
  Thus $\chi'$ anti-distinguishes $\phi$ and $\psi$.
\end{proof}

The following lemma is the abstract content of the PBR theorem, and holds for any concrete monoidal operational model.

\begin{lemma}\label{lem:abstractpbr}
  A concrete operational model $F \colon \cat{C} \to \cat{BoRel}$ is ontic as soon as there are states $\phi,\psi \in \cat{C}(I,A)$ and a natural number $n>0$ for which 
  $\{ \phi^{\otimes n} \otimes \phi^{\otimes n}, \phi^{\otimes n} \otimes \psi^{\otimes n}, \psi^{\otimes n} \otimes \phi^{\otimes n}, \psi^{\otimes n} \otimes \psi^{\otimes n} \}$ is anti-distinguishable.
\end{lemma}
\begin{proof}
  By Lemmas~\ref{lem:preservedistinguishability} and~\ref{lem:preserventensordistinguishable}, the set
  $\{ F(\phi)^{\otimes n} \otimes F(\phi)^{\otimes n}, F(\phi)^{\otimes n} \otimes F(\psi)^{\otimes n}, F(\psi)^{\otimes n} \otimes F(\phi)^{\otimes n}, F(\psi)^{\otimes n} \otimes F(\psi)^{\otimes n} \}$ is anti-distinguishable in $\cat{BoRel}$.
  By Lemma~\ref{lem:twotensordistinguishablesrel}, therefore the set $\{F(\phi)^{\otimes n}, F(\psi)^{\otimes n}\}$ is anti-distinguishable. 
  Lemma~\ref{lem:ntensordistinguishablesrel} now guarantees that $F(\phi)$ and $F(\psi)$ are anti-distinguishable.
  Hence $F$ is ontic.
\end{proof}

We can now finish the proof of our categorical analogue of the PBR theorem by constructing specific states in $\cat{FHilb}$.

\begin{theorem}\label{thm:categoricalpbr}
  Any monoidal operational model $\cat{FHilb} \to \cat{BoRel}$ is ontic.
\end{theorem}
\begin{proof}
  It suffices to satisfy the hypotheses of Lemma~\ref{lem:abstractpbr} for $\cat{FHilb}$.
  As explained in the proof sketch of Theorem~\ref{thm:pbrproper} in Appendix~\ref{sec:pbr}, there is an anti-distinguishing measurement when $\psi = \ket{0}$ and $\phi = \ket{+}$, according to equations~\eqref{eqn:measurement1}--\eqref{eqn:measurement4} in Appendix~\ref{sec:pbr}. Furthermore, for any pair of states $\psi, \phi$ there exists $n > 0$ such that $|\braket{\psi^{\otimes n} | \phi^{\otimes n}}| \leq \braket{0 | +} = 1/\sqrt{2}$, because $|\braket{\psi^{\otimes n} | \phi^{\otimes n}}| = |\braket{\psi| \phi}|^n$ and $|\braket{\psi| \phi}| < 1$. By applying a unitary if necessary, we may assume without loss of generality that $\braket{\psi \mid \phi}$ is real.

  Following~\cite{leifer:review}, write $\braket{\psi^{\otimes n}| \phi^{\otimes n}} = \gamma$, and consider the completely positive map
  $
    \mathcal{E} \colon \rho \mapsto K_0 \rho K^{\dagger}_0 + K_1 \rho K^{\dagger}_1
  $
  with Krauss operators
  \begin{align*}
  K_0 &= \ket{0}\bra{0} + \tan{\gamma} \ket{1} \bra{1}\text, &
  K_1 &= \left( \sqrt{\frac{1 - \tan^2{\gamma}}{2}} \right) ( \ket{0} + \ket{1}) \bra{1}\text.
  \end{align*}
  As is shown in \cite{leifer:review}, up to global phases:
  \[
  \mathcal{E}(\ket{\psi^{\otimes n}}\bra{\psi^{\otimes n}}) = \ket{0}\bra{0} \quad \quad
  \mathcal{E}(\ket{\phi^{\otimes n}}\bra{\phi^{\otimes n}}) = \ket{+}\bra{+}
  \]
  Applying $\mathcal{E}$ leads to a pair of states that are anti-distinguishable. Of course, being a completely positive linear map, $\mathcal{E}$ is itself not a morphism in $\cat{FHilb}$. However, by Stinespring dilation we can always obtain a morphism in $\cat{FHilb}$ by first tensoring an ancilla system to our state. We are allowed to do this because the functor is monoidal. Then using the measurement given by the projections~\eqref{eqn:measurement1}--\eqref{eqn:measurement4}, it follows that $\{ \psi^{\otimes n} \otimes \psi^{\otimes n}, \psi^{\otimes n} \otimes \phi^{\otimes n}, \phi^{\otimes n} \otimes \psi^{\otimes n}, \phi^{\otimes n} \otimes \phi^{\otimes n} \}$ is anti-distinguishable. (A different anti-distinguishing measurement is considered in~\cite{puseybarrettrudolph:reality}.)
\end{proof}

There are a two important differences between Theorems~\ref{thm:categoricalpbr} and the original PBR theorem (Theorem~\ref{thm:pbrproper}). First, operational models are more restrictive than ontological theories due to the mapping between categories being a functor. Second, the cartesian product assumption of the preparation independence postulate (Definition~\ref{def:pip}), only asks that the set of product states map to a product space in the ontological theory. However, the requirement that a functor is monoidal forces tensor products of Hilbert spaces to be mapped to products of measurable spaces. Dealing with product states in particular may be modelled by precomposing with the identity functor $(\cat{FHilb},\oplus) \to (\cat{FHilb},\otimes)$ that is oplax monoidal by $\phi \times \psi \mapsto \phi \otimes \psi$.

\section{Duality-preserving operational models}\label{sec:dagger}

Next we consider operational models that respect a duality between states and effects.\footnote{The terminology ``duality'' is not ideal, because in general not every measurement may be induced by a state, as is the case in quantum theory. But ``state-induced-measurement-preserving operational model'' is a mouthful.}
In $\cat{FHilb}$, any state $\psi \colon I \to A$ induces a measurement $\chi \colon A \to \mathbb{C}^2$ via the Born rule $\chi(a) = (p,1-p)$ for $p=|\braket{a|\psi}|^2$.\footnote{Although $\cat{FHilb}$ is a dagger category, this state and the measurement are not each other's dagger, as their type mismatches. However, note that the measurement is essentially derived from the effect associated by the dagger to the state. This point of view is not uncommon in categorical quantum foundations~\cite{jacobsetal:effectus,selbyscandolocoecke:reconstructing}.}
In $\cat{BoRel}$, any state $\mu$ on a measurable space $\Lambda$ induces a measurement $\chi \colon \Lambda \to [0,1]$ via evaluation $\chi(\lambda) = \mu \{\lambda\}$. This requires singletons $\{\lambda\}$ to be measurable sets. This holds for discrete $\sigma$-algebras and Borel $\sigma$-algebras and hence is fine when working with probability distributions or measures on topological spaces.
Both $\cat{FHilb}$ and $\cat{BoRel}$ thus canonically preserve duality in the following sense.

\begin{definition}\label{def:dagger}
  A \emph{state-measurement duality} on an operational category is a family of functions $\dag \colon \cat{C}(I,A) \to \cat{C}(A,2)$.
  A \emph{duality-preserving operational model}	is an operational model between operational categories with state-measurement duality that preserves the duality: $F(\psi^\dag) = F(\psi)^\dag$.
\end{definition}

\begin{proposition}
  There does not exist a duality-preserving operational model $\cat{FHilb} \to \cat{BoRel}$.
\end{proposition}
\begin{proof}
  By definition, for two states $\psi,\psi' \colon \mathbb{C} \to A$ in $\cat{FHilb}$, their overlap is $\braket{\psi^\dag \circ \psi'}$.
  A duality-preserving operational model $F$ must satisfy
  $
    \braket{ \psi^\dag \circ \psi'}_{\cat{FHilb}} = \braket{ F(\psi)^\dag \circ F(\psi')}_{\cat{BoRel}}
  $.
  In other words, an operational model is duality-preserving exactly when it is maximally epistemic. But any maximally epistemic ontological theory must be noncontextual and outcome deterministic~\cite{leifer2013maximally}, which contradicts the contextuality of quantum theory. Hence such a model cannot exist.
\end{proof}

\section{Equivariant operational models}\label{sec:equivariant}

In this section, we consider another property of operational models: \emph{symmetry}.
This property forces the operational model to adhere to the Schr{\"o}dinger equation.
In the categorical setting, we can naturally express this as \emph{equivariance} under all quantum maps\footnote{There are many categorical notions of action, a mere functor being the least structured one. We prefer to write our definition of equivariance below explicitly, rather than phrase it as one of these notions.}.
The simple fact of functoriality of operational models will show that equivariance implies symmetry, and so, following Aaronson et al~\cite{aaronson2013psi}, rules out maximally nontrivial operational models. 
To phrase these properties there is a price to be paid, namely that there must be a connection between the ontic spaces and the dynamics.

\begin{definition}
  Let $F \colon \cat{FHilb} \to \cat{BoRel}$ be a concrete operational model.
  An \emph{action} is a map $\cat{FHilb}(A,B) \times \Sigma_{F(A)} \to \Sigma_{F(B)}$ for each $A,B \in \cat{FHilb}$, that turns a bounded linear map $f \colon A \to B$ and a measurable set $U \in \Sigma_{F(A)}$ into a measurable set $f \cdot U \in \Sigma_{F(B)}$, satisfying $(g \circ f) \cdot U = g \cdot (f \cdot U)$ and $\id \cdot U = U$. 
  The operational model is \emph{equivariant} when 
  \begin{equation}\label{eq:equivariance}
    F(f \circ \psi)(*,U) = F(\psi)(*, f\cdot U)
  \end{equation}
  for all $\psi \colon \mathbb{C} \to A$ and $f \colon A \to A$ in $\cat{FHilb}$ and $U \in \Sigma_{F(A)}$.
\end{definition}

In general, the action $U \mapsto f \cdot U$ can be an arbitrary function $\Sigma_A \to \Sigma_B$, as long as it is compositional as in the definition above.
When we demand that $F(A)$ is $A$ itself under its discrete $\sigma$-algebra, there is a canonical action $\cat{FHilb}(A,B) \times \Sigma_A \to \Sigma_B$ given by $f \cdot U = \{f(\lambda) \mid \lambda \in U\}$. This requirement, that ontic states are simply quantum states, is not particularly strong. The model of Example~\ref{ex:universal} satisfies it (although it is not an operational model). Moreover, this functor is \emph{universal}, in that any concrete operational model $\cat{FHilb}\to \cat{BoRel}$ must factor uniquely through the functor into the pertinent subcategory of $\cat{BoRel}$, that assigns the discrete $\sigma$-algebra to $A$ itself and turns a completely positive map $f \colon A \to B$ into the Markov kernel $(a,V) \mapsto \chi_V(f(a))$. 

\begin{proposition}
  There is no maximally nontrivial equivariant operational model $\cat{FHilb} \to \cat{BoRel}$.
\end{proposition}
\begin{proof}
  By considering the action on pure states $\psi \in A$ and unitary evolutions $f : A \to A$, $f(a)=u a$ for some unitary $u$, one recovers from~\eqref{eq:equivariance} the property of symmetry as stated in~\cite{aaronson2013psi}:
  \[
    \mu_{u\psi}(\lambda) = \mu_\psi(u \lambda)\text,
  \]
  where $\mu_\psi$ denotes the measure $F(\psi)$ induced by the state $\psi$. The result follows from~\cite{aaronson2013psi}.
\end{proof}

\section{Signed operational models} \label{sec:signed}

We've seen severe limitations on the kinds of functors $\cat{FHilb} \to \cat{BoRel}$ allowed. What if we change the target category? Ideally not too much, to retain some notion resembling probability theory.
To do so, note that, as in Section~\ref{sec:monoidal}, one of the fundamental differences between $\cat{FHilb}$ and $\cat{BoRel}$ is that states that are not anti-distinguishable in $\cat{BoRel}$ cannot become anti-distinguishable under finite tensor products. This is not the case with quantum states. What makes the two situations so different? Lemmas~\ref{lem:twotensordistinguishablesrel} and~\ref{lem:ntensordistinguishablesrel} highlight that the distinction stems from the fact that classical probability distributions must be positive and obey the Kolmogorov sum rule. In contrast, quantum states are represented as vectors of amplitudes that can be both positive and negative (and in general even complex). Consequently quantum states are subject to interference, which cannot be reproduced with classical probability distributions.
A natural idea is to allow measures to take both positive and negative values, leading to a category whose objects are still measureable spaces, while morphisms are \emph{signed Markov kernels}.
(An alternative solution, that still uses positive measures, is to abandon the Kolmogorov sum rule. This approach leads to \emph{quantum measures}, but presents a number of challenges discussed in Appendix~\ref{sec:quantummeasure}.)

\begin{definition}
  A signed Markov kernel from a measurable space $(X,\Sigma_X)$ to a measurable space $(Y,\Sigma_Y)$ is a function $f \colon X \times \Sigma_Y \to [-1,1]$ such that $f(-,V) \colon X \to [-1,1]$ is a bounded measurable function for each $V \in \Sigma_Y$, and $f(x,-) \colon \Sigma_Y \to [-1,1]$ is a signed measure for each $x \in X$. Signed measures have the same properties as unsigned measures, except that they can be both positive and negative. In particular, they still obey the Kolmogorov sum rule for disjoint sets. We will require that a signed measure is normalised: for all $x \in X$ it should be the case that $f(x, Y) = 1$.
\end{definition}

\begin{proposition}
  Measurable spaces and signed Markov kernels form a category $\cat{QSRel}$ with composition $(g \circ f)(x,W) = \int g(y,W) f(x,\mathrm{d}y)$, that is monoidal under $f \otimes f' \colon ((x,x'),V \times V')\mapsto f(x,V) \cdot f'(x',V')$.
  Borel spaces form a full subcategory $\cat{QBoRel}$ that inherits the monoidal structure.
\end{proposition}
The category is called $\cat{QSRel}$ because probability distributions arising from signed measures are referred to as \emph{quasiprobability distributions}.
\begin{proof}
  The proof of e.g.~\cite[Proposition~3.2]{panangaden:srel} goes through nearly verbatim. Only the very last step, using the monotone convergence theorem, has to be amended: split the sequence into a positive part and a negative part using Hahn decomposition, apply monotone convergence to both, and then combine them again by subtracting the negative from the positive.
  Alternatively, one can realise that a Giry-like monad~\cite{giry}, that assigns to a measureable space its set of signed measures, is still well-defined and monoidal, and $\cat{QSRel}$ is its Kleisli category.
\end{proof}

We will now show that we \emph{can} have a monoidal operational model. The idea of having a quasiprobabilistic interpretation of quantum mechanics has been considered before, notably in~\cite{ferrie2011quasi}.
One possible construction that is physically motivated relies on \emph{Wigner functions}~\cite{PhysRev.40.749}. A Wigner function of a quantum state is a quasiprobability distribution over the phase space\footnote{Loosely speaking, phase space is the space of all possible position and momenta for a quantum system. See~\cite{gibbons2004discrete} for a description of the concept in the case of finite-dimensional systems.} associated to that space. This construction, particularly for the case of finite-dimensional Hilbert spaces, is described in~\cite{gibbons2004discrete, gross2006hudson}, whose approach we follow.
Before doing so, we consider one further strengthening of the result. Up to this point, our source category in the operational model has been the category \cat{FHilb}. For this result, we will instead consider the category $\CPs[\cat{FHilb}]$ of finite-dimensional $C^*$ algebras and CP maps, corresponding to ``mixed-state quantum mechanics'' (in contrast to \cat{FHilb} which models ``pure-state quantum mechanics''). Just like \cat{FHilb}, this is still an operational category having $\mathbb{C}^2$ as the distinguishing object, $\Omega=[0,1]$ as the set of abstract probabilities and the Born rule as the evaluation map.

One could ask why we haven't used this category for the previous results. The reason is that \emph{any} operational model having $\CPs[\cat{FHilb}]$ as the source category is \emph{preparation non-contextual} \cite{spekkens2005contextuality}. This means that the measure over ontic states that is assigned to a density matrix is independent of the ensemble that produced that density matrix. A functor naturally enforces this condition. However, as is shown in \cite{spekkens2005contextuality}, there does not exist a preparation non-contextual ontological model for quantum mechanics. Thus, operational models $\CPs[\cat{FHilb}] \to \cat{BoRel}$ are implicitly ruled out.
Our next result shows that if the target category is \cat{QBoRel}, then it is in fact possible to have an operational model even if this model is preparation non-contextual\footnote{A similar result was shown in~\cite{van2017quantum} for the case when the target category consists of finite sets and stochastic maps.}.

\begin{theorem}\label{thm:signed}
  There is an epistemic operational model $\CPs[\cat{FHilb}] \to \cat{QBoRel}$. It is in fact maximally epistemic up to a factor, in that~\eqref{eq:maximallyepistemic} holds up to a multiplicative constant.
\end{theorem}
\begin{proof}
  We may without loss of generality restrict to C*-algebras of odd dimension, by considering the functor $\CPs[\cat{FHilb}] \to \CPs[\cat{FHilb}]$ that maps C*-algebras of odd dimension to themselves, C*-algebras $A$ of even dimension to $A \oplus \mathbb{C}$; on morphisms, it embeds a completely map into the top-left corner of a block matrix whose other entries are zero. 
  Furthermore, extending by direct sums, we may restrict to C*-algebras $\M_n$ of all $n$-by-$n$ matrices for odd $n$.

  Our construction starts with, for each odd number $n$, a family $\Lambda_n$ of $n^2$ many $n$-by-$n$ matrices with the following properties:
  \begin{enumerate}
    \item $\sigma=\sigma^\dag$ for each $\sigma \in \Lambda_n$;
    \item $\Tr(\sigma)=1$ and $\sigma^2=1$ for each $\sigma \in \Lambda_n$;
    \item $\Tr(\sigma \tau)=0$ for distinct $\sigma,\tau \in \Lambda_n$.
  \end{enumerate}
  We will see later that such a family indeed exists.
  Observe that $\Lambda_n$ is an orthonormal basis for $\M_n$ under the Hilbert-Schmidt inner product.
  We may read the fact that any completely positive map $f \colon \M_m \to \M_n$ is completely determined by its action on $\Lambda_m$, as saying that the effect of a quantum channel is completely determined by how it acts on a tomographically complete set of observables.
  Thus it is completely determined by its \emph{transfer matrix} $f_{ij} = \Tr(\sigma_i^n f(\sigma_j^m))/m$. Note that the dimension of the transfer matrix is $n^2 \times m^2$.

  Suppose for a minute that $f \colon \M_m \to \M_n$ implements a function $f' \colon \{1,\ldots,m^2\}\to\{1,\ldots,n^2\}$ via $f_{ij} = \delta_{f'(i),j}$; we call such a map $f$ a \emph{point channel}.
  If $g \colon \M_n \to \M_p$ is another point channel, it follows from properties 2 and 3 that:
  \begin{align*}
	(g \circ f)_{ij} 
	& = \Tr(\sigma^p_j \sigma^p_{g'(f'(i))})/p 
	  = \delta_{j,g'(f'(i))} \\
	& = \sum_{k=1}^{n^2} \delta_{j,g'(k)} \delta_{k,f'(i)} 
	  = \frac{1}{np} \sum_{k=1}^{n^2} \Tr(\sigma_j^p \sigma^p_{f'(k)}) \Tr(\sigma_k^m \sigma^m_{f'(i)}) 
	  = \sum_{k=1}^{n^2} g_{ik} f_{kj}\text.
  \end{align*}
  Because any channel $f$ is a normalised linear combination of point channels, it follows that the matrix of $g \circ f$ of a composition is the multiplication of the matrices of $g$ and $f$.

  Next, the fact that $\Lambda_n$ is an orthonormal basis implies that any density matrix $\rho \in \M_n$ is determined by the coefficients 
  $
    v_i(\rho)=\Tr(\rho \sigma_i)/n
  $
  as $\rho=\sum v_i(\rho) \sigma_i$.
  Because $Tr(\rho \sigma_i)$ is the expectation value of $\sigma_i$ when measuring the state $\rho$, in fact $-n \leq v_i(\rho) \leq n$.
  We will regard the normalised vector $v(\rho)=(v_1(\rho),\ldots,v_n(\rho))/n$ as the quasiprobability distribution associated to the state $\rho$, and more generally the normalised matrix $f_{ij}$ as the stochastic map associated to the channel $f$.
  That is, the functor $F \colon \CPs[\cat{FHilb}] \to \cat{QSRel}$ sends an object $\M_n$ to the set $\Lambda_n$ under the discrete $\sigma$-algebra, and it sends a morphism $f \colon \M_m \to \M_n$ to the signed Markov kernel $Ff \colon \Lambda_m \times \Sigma_{\Lambda_n} \to [-1,1]$ given by
  \[
    Ff (\sigma, W) = \sum_{k \in W} ((f_{ij}) v(\sigma))_k
    \text,
  \]
  where for ease of notation we pretended that $W$ contained indices of matrices from $\Lambda_n$ rather than the matrices themselves.
  Intuitively, the map $Ff$ takes an input quasiprobability vector, applies the  transfer matrix $f_{ij}$, and thus obtains an output quasiprobability vector.
  
  Observe that indeed $f(-, W) \colon \Lambda_m \rightarrow [-1, 1]$ is a bounded measurable function, and that indeed $f(\sigma, -) \colon \Sigma_{\Lambda_n} \to [-1, 1]$ is a signed measure. Condition $2$ ensures that the quasiprobability distributions are normalised. Thus $F$ is well-defined.
  It is a functor, because, as we have seen, composition is preserved when moving from morphisms to their transfer matrices.

  We interpret this functor as assigning Wigner functions to quantum states. Specifically, the quasiprobability vector $v(\rho)$ is the Wigner function associated to the state $\rho$. These quasiprobability distributions are defined over phase space. The operators $\sigma \in \Lambda_n$ are known as \emph{phase space point operators} and are observables associated to each point in phase space. Describing them in detail is beyond the scope of this paper, and we refer the reader to~\cite{gibbons2004discrete, gross2006hudson} for the appropriate details. For our construction, it is sufficient that such operators exist and satisfy conditions 1--3.

  It remains to show that the operational model $F$ is maximally epistemic up to a factor: that the overlap of quantum states matches that of their Wigner functions up to a multiplicative constant.
  Consider two density matrices $\rho, \tau \in \mathbb{M}_n$. Their trace distance is given by
  \[
    \frac{1}{2} Tr(|\rho - \tau|)
    = \frac{1}{2} \sum_{i=1}^{n^2} | v_i(\rho) - v_i(\tau) | Tr(|\sigma_i|)\text.
  \]
  Now condition $2$ implies $Tr(|\sigma_i|) = n$, and therefore
  the trace distance is
  $
    \frac{n}{2} \sum_{i=1}^{n^2} | v_i(\rho) - v_i(\tau) | 
  $.
  But this equals the variation distance between the two Wigner functions. It follows that whenever two quantum states have nontrivial overlap, their Wigner functions will also have nontrivial overlap.
\end{proof}

\begin{remark}
  The epistemic operational model of Theorem~\ref{thm:signed} is in fact monoidal when restricting to the subcategory of $\CPs[\cat{FHilb}]$ of odd-dimensional C*-algebras (and taking $\mathbb{C}^3$ as distinguishing object).
\end{remark}
\begin{proof}
  The tensor product of two completely positive maps $f_1\colon \M_{m_1} \to \M_{n_1}$ and $f_2 \colon \M_{m_2} \to \M_{n_2}$ is $f_1 \otimes f_2 \colon \M_{m_1m_2} \to \M_{n_1n_2}$. For odd $m$ and $n$, the operators of $\Lambda_m \otimes \Lambda_n$ again satisfy properties 1--3 in the proof of Theorem~\ref{thm:signed}. Bilinearity of the tensor product thus shows that the transfer matrix of $f_1 \otimes f_2$ is the tensor product of the transfer matrices of $f_1$ and $f_2$.
\end{proof}

\section*{Acknowledgements}
We thank Matty Hoban and Petros Wallden for useful discussions.
We are also grateful to the anonymous referees for useful comments and suggestions.
 Alexandru Gheorghiu is supported by MURI Grant FA9550-18-1-0161 and the IQIM, an NSF Physics Frontiers Center (NSF Grant PHY-1125565) with support of the Gordon and Betty Moore Foundation (GBMF-12500028). Chris Heunen is supported by EPSRC Fellowship EP/R044759/1.

\bibliographystyle{eptcs}
\bibliography{pbr}

\appendix 
\section{The PBR theorem} \label{sec:pbr}

This appendix briefly recalls the PBR theorem and its assumptions.
The PBR theorem shows that there cannot be an ontological model of quantum mechanics that is epistemic and simultaneously satisfies a property known as preparation independence~~\cite{leifer:review}:

\begin{definition}\label{def:pip}
  An ontological theory of quantum mechanics satisfies the \emph{preparation independence postulate (PIP)}, if given Hilbert spaces $H_1$ and $H_2$, with associated ontic spaces $\Lambda_1$ and $\Lambda_2$, the following are satisfied:
  \begin{description}
    \item[Cartesian product assumption (CPA)] 
      Let $\{ \ket{\psi} \otimes \ket{\phi} \mid \ket{\psi} \in H_1, \ket{\phi} \in H_2 \}$ be the set of all product states in $H_1 \otimes H_2$. Its ontic space is $\Lambda_1 \times \Lambda_2$, with $\sigma$-algebra $\Sigma_{\Lambda_1} \otimes \Sigma_{\Lambda_2}$.
    \item[No-correlation assumption (NCA)] 
      Given states $\ket{\psi} \in H_1$ and $\ket{\phi} \in H_2$, with associated measures $\mu_\psi \colon \Sigma_{\Lambda_1} \to [0,1]$ and $\mu_\phi \colon \Sigma_{\Lambda_2} \to [0,1]$, the measure  $\mu_{\psi \otimes \phi} \colon \Sigma_{\Lambda_1} \otimes \Sigma_{\Lambda_2} \to [0,1]$ associated to $\ket{\psi} \otimes \ket{\phi}$ is determined by
	  \[
	    \mu_{\psi \otimes \phi}(U \times V) = \mu_{\psi}(U) \mu_{\phi}(V)
	  \]
      for all $U \in \Sigma_{\Lambda_1}$ and $V \in \Sigma_{\Lambda_2}$\footnote{Note that for sets $Z \in \Sigma_{\Lambda_1} \otimes \Sigma_{\Lambda_2}$ that cannot be expressed as $Z = U \times V$ with $U \in \Sigma_{\Lambda_1}$ and $V \in \Sigma_{\Lambda_2}$, we simply use the Kolmogorov sum rule, which states that for any disjoint $U,V \in \Sigma$ and any measure $\mu \colon \Sigma \to [0,1]$ it is the case that $\mu(U \cup V) = \mu(U) + \mu(V)$. Specifically, any set $Z \in \Sigma_{\Lambda_1} \otimes \Sigma_{\Lambda_2}$ can be expressed as a countable union of (cartesian) products of sets from $\Sigma_{\Lambda_1}$ and $\Sigma_{\Lambda_2}$. Knowing the value of $\mu_{\psi \otimes \phi}$ on those sets and using the Kolmogorov sum rule determines its value on any set $Z$.}.
  \end{description}
\end{definition}

\begin{theorem}[PBR \cite{puseybarrettrudolph:reality}] \label{thm:pbrproper}
  Any ontological theory of quantum mechanics satisfying PIP must be ontic.
\end{theorem}
\begin{proof}[Proof sketch]
  We will merely sketch the general ideas behind the proof. For the full proof see \cite{puseybarrettrudolph:reality} or \cite{leifer:review}.
  Suppose we have an epistemic ontological theory of quantum mechanics. This means that there exists a pair of distinct quantum states that have overlapping support in their distributions over ontic states. If $\ket{0}$ and $\ket{+}$ are two such states, then
  \begin{equation} \label{ineq:overlap}
    D(\mu_0, \mu_+) < 1\text.
  \end{equation}
  Now suppose we have two agents, Alice and Bob, that independently prepare either the $\ket{0}$ or the $\ket{+}$ states and send them to a third agent, Eve. The states that Eve can possibly receive are $\ket{0}_A \ket{0}_B$, $\ket{0}_A \ket{+}_B$, $\ket{+}_A \ket{0}_B$, and $\ket{+}_A \ket{+}_B$.
  By the preparation independence postulate, the ontic space associated to these states is the product of the ontic spaces for Alice and Bob's states and the measures will be product measures.
  Using subadditivity of the variational distance, together with inequality~\eqref{ineq:overlap} and the fact that $D(\mu, \nu) = D(\nu, \mu)$, for any measures $\mu$ and $\nu$:
  \begin{align} \label{ineq:productoverlap1}
	D(\mu_{00}, \mu_{0+}) <& 1 & D(\mu_{00}, \mu_{+0}) <& 1 \\
  \label{ineq:productoverlap2}
	D(\mu_{++}, \mu_{0+}) <& 1 & D(\mu_{++}, \mu_{+0}) <& 1
  \end{align}
  What about $D(\mu_{00}, \mu_{++})$? From inequality~\eqref{ineq:overlap}, we conclude that the measures associated to $\ket{0}$ and $\ket{+}$ have nontrivial overlap. But given that the measures associated to $\ket{0}_A \ket{0}_B$ and $\ket{+}_A \ket{+}_B$ are product measures, this means that they will also have nontrivial overlap\footnote{Another way of saying this is that if $D(\mu, \nu) < 1$ then $D(\mu \times \mu, \nu \times \nu) < 1$, for any two measures $\mu$ and $\nu$.}:
  \begin{equation} \label{ineq:productoverlap3}
    D(\mu_{00}, \mu_{++}) < 1
  \end{equation}
  Now inequalities~\eqref{ineq:productoverlap1}--\eqref{ineq:productoverlap3} provide a subset of ontic states $\Delta \in \Sigma_{AB}$, $\Delta \neq \emptyset$, such that $\mu_{00}(\Delta)$, $\mu_{0+}(\Delta)$, $\mu_{+0}(\Delta)$, and $\mu_{++}(\Delta)$ are all strictly positive. Figure~\ref{fig:pbrimg} illustrates this fact.

  \begin{figure}[htbp!]
    \centering
    \includegraphics[scale=0.25]{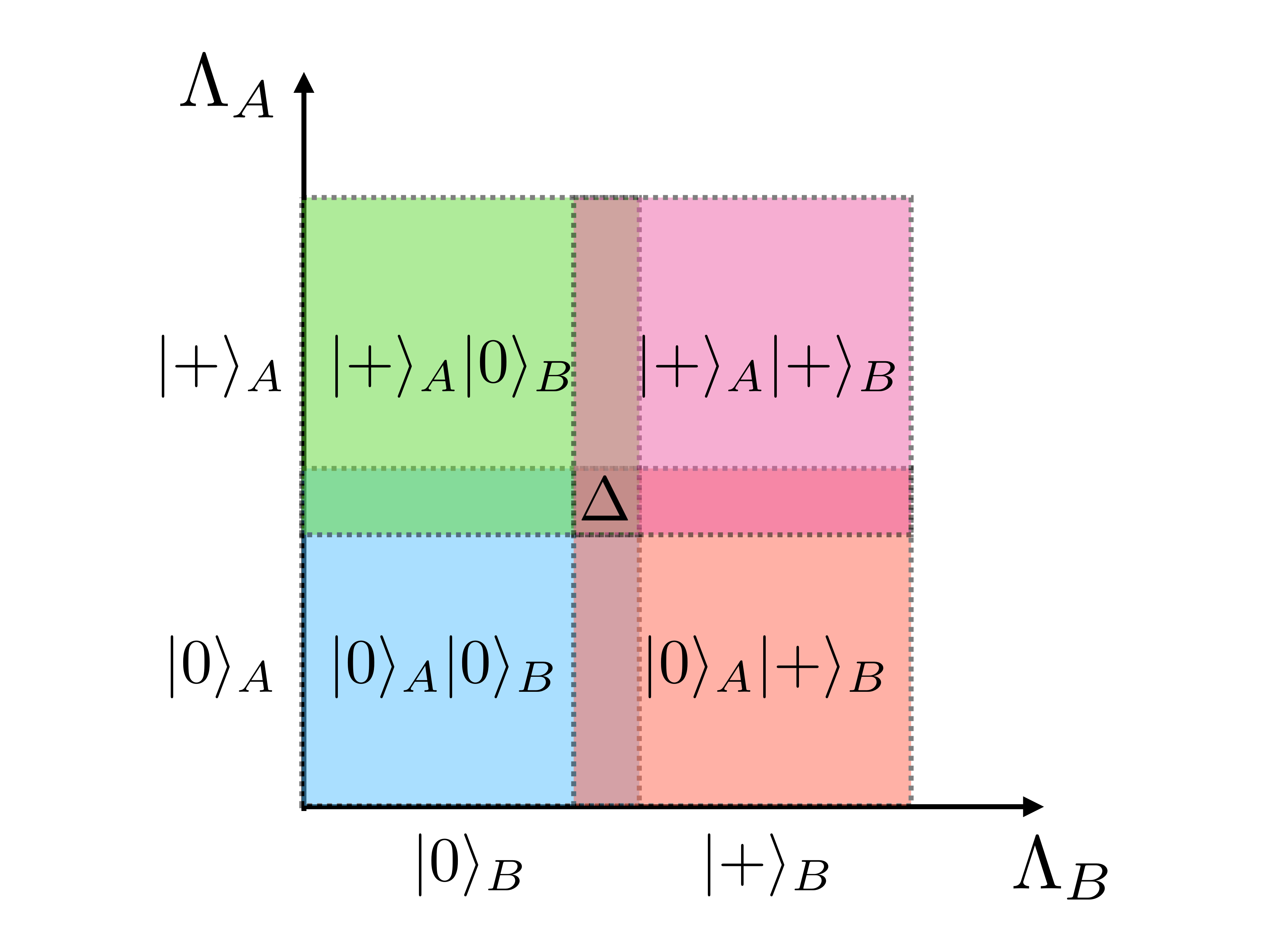}
    \caption{Schematic illustration of the ontic space. The coloured regions represent sets in ontic space on which the distributions have non-zero support. For instance, the blue region corresponds to those ontic states on which $\mu_{00}$ has non-zero support. Notice that all four states overlap in the middle region, $\Delta$.}
    \label{fig:pbrimg}
  \end{figure}

  Suppose now that Eve performs a projective measurement on her two qubits, defined by the following two-qubit basis vectors: 
  \begin{align} 
    \label{eqn:measurement1}
	\ket{\chi_1} &= \frac{1}{\sqrt{2}} ( \ket{0}_A \ket{1}_B + \ket{1}_A \ket{0}_B ) \\
	\label{eqn:measurement2}
	\ket{\chi_2} &= \frac{1}{\sqrt{2}} ( \ket{0}_A \ket{-}_B + \ket{1}_A \ket{+}_B ) \\
	\label{eqn:measurement3}
	\ket{\chi_3} &= \frac{1}{\sqrt{2}} ( \ket{+}_A \ket{1}_B + \ket{-}_A \ket{0}_B ) \\
	\label{eqn:measurement4}
	\ket{\chi_4} &= \frac{1}{\sqrt{2}} ( \ket{+}_A \ket{-}_B + \ket{-}_A \ket{+}_B )
  \end{align}
  Because we are considering an epistemic model, we will associate a response function to each of these outcomes and denote them $\xi_1, \xi_2, \xi_3, \xi_4 \colon \Lambda_{AB} \rightarrow [0,1]$. 
  Now
  \[
    \braket{00|\chi_1} = \braket{0+|\chi_2} = \braket{+0|\chi_3} = \braket{++|\chi_4} = 0\text.
  \]
  In other words, whatever outcome Eve obtains from her measurement, it will certainly rule out one of the four possible states that she received.
  But then:
  \begin{align*}
	\int_{\Lambda_{AB}} \xi_1(\lambda) d\mu_{00}(\lambda) &= 0 
	&
	\int_{\Lambda_{AB}} \xi_2(\lambda) d\mu_{0+}(\lambda) &= 0 
	\\
	\int_{\Lambda_{AB}} \xi_3(\lambda) d\mu_{+0}(\lambda) &= 0 
	&
	\int_{\Lambda_{AB}} \xi_4(\lambda) d\mu_{++}(\lambda) &= 0
  \end{align*}
  Since there exists a non-trivial $\Delta$ such that $\mu_{00}(\Delta), \mu_{0+}(\Delta), \mu_{+0}(\Delta), \mu_{++}(\Delta) > 0$, it must be the case that for all $k \in \{1,2,3,4\}$ and for all $\lambda \in \Delta$, $\xi_k(\lambda) = 0$. However, this contradicts the fact that 
  \begin{equation*}
    \sum\limits_{k=1}^4 \xi_k(\lambda) = 1
  \end{equation*}
  for all $\lambda \in \Lambda_{AB}$.

  This argument assumed that the states having non-trivial overlap in ontic space are $\ket{0}$ and $\ket{+}$. PBR showed that the above argument can be generalized for any pair of states $\ket{\psi}$ and $\ket{\phi}$ with $0 < |\braket{\psi\mid\phi}| < 1$. In fact, in the simple proof given above, the only place that explicitly used $\ket{0}$ and $\ket{+}$ was to define Eve's entangled measurement. The generalization of PBR consists in showing that Eve can always construct such an entangled measurement for $n$-fold tensor products of $\ket{\psi}$ and $\ket{\phi}$ if $n$ is sufficiently large. The intuition for this, as explained in \cite{leifer:review}, is the following. Suppose one considers two states $\ket{\psi}$ and $\ket{\phi}$ such that $|\braket{\psi \mid \phi}| < 1$. Clearly, there exists an $n > 0$ such that $|\braket{\psi^{\otimes n} \mid \phi^{\otimes n}}| \leq 1/\sqrt{2}$. We also know that $\braket{0\mid +} = 1/\sqrt{2}$. If the inner product between $\ket{\psi}^{\otimes n}$ and $\ket{\phi}^{\otimes n}$ is at most that between $\ket{0}$ and $\ket{+}$, that is, if $\ket{\psi}^{\otimes n}$, $\ket{\phi}^{\otimes n}$ are at least ``as distinguishable'' as $\ket{0}$ and $\ket{+}$, then there exists a mapping from $\ket{\psi}^{\otimes n}$ to $\ket{0}$ and from $\ket{\phi}^{\otimes n}$ to $\ket{+}$
  . Eve can then perform the previously described anti-distinguishing measurement.
\end{proof}

\section{Quantum measures} \label{sec:quantummeasure}

Section~\ref{sec:signed} showed one way to evade the PBR obstruction: replace ordinary probability measures with signed measures. Another option is to allow positive measures that can violate the Kolmogorov sum rule. These measures should not be completely unconstrained, and should still reproduce the Hilbert space inner product. There is a natural candidate that satisfies these properties, namely \emph{quantum measures} (sometimes also called \emph{quantal measures}) \cite{sorkin1994quantum, sorkin1995quantum, salgado2002some, martin2005random, surya2010quantum, dowker2010extending}. We give two equivalent definitions of quantum measures, taken from \cite{dowker2010extending}:

\begin{definition}
Let $\Lambda$ be a measurable space with associated $\sigma$-algebra $\Sigma_{\Lambda}$. A \emph{quantum measure} over $\Sigma_{\Lambda}$ is a function $\mu \colon \Sigma_{\Lambda} \rightarrow [0,1]$ satisfying the following properties:
\begin{itemize}
\item \textbf{Positivity}. For all $U \in \Sigma_{\Lambda}$, $\mu(U) \geq 0$;
\item \textbf{Normalisation}. $\mu(\Lambda) = 1$;
\item \textbf{Quantum sum rule}. For all pairwise disjoint sets $U,V,W \in \Sigma_{\Lambda}$:
\begin{equation*}
\mu(U \cup V \cup W) = \mu(U \cup V) + \mu(U \cup W) + \mu(V \cup W) - \mu(U) - \mu(V) - \mu(W)\text.
\end{equation*}
\end{itemize}
\end{definition}

An equivalent characterization uses \emph{decoherence functionals}, first considered in \cite{dowker2010hilbert}, which yields a natural notion of inner product over a $\sigma$-algebra.

\begin{definition}
Let $\Lambda$ be a measurable space with associated $\sigma$-algebra $\Sigma_{\Lambda}$. A \emph{decoherence functional} is a function $\mathcal{D} \colon \Sigma_{\Lambda} \times \Sigma_{\Lambda} \rightarrow \mathbb{C}$ satisfying the following properties:
\begin{itemize}
\item \textbf{Hermitian}. For all $U,V \in \Sigma_{\Lambda}$, $\mathcal{D}(U,V) = \mathcal{D}(V,U)^*$;
\item \textbf{Normalisation}. $D(\Lambda, \Lambda) = 1$;
\item \textbf{Finite bi-additivity}. For all $U \in \Sigma_{\Lambda}$ and all mutually disjoint sets $V_1,\ldots,V_n \in \Sigma_{\Lambda}$:
\begin{equation*}
\mathcal{D}(U, \bigcup_{i=1}^n V_i) = \sum\limits_{i=1}^m \mathcal{D}(U, V_i)
\end{equation*}
Similarly, for all $V \in \Sigma_{\Lambda}$ and all mutually disjoint sets $U_1,\ldots,U_n \in \Sigma_{\Lambda}$:
\begin{equation*}
\mathcal{D}(\bigcup_{i=1}^n U_i, V) = \sum\limits_{i=1}^m \mathcal{D}(U_i, V)
\end{equation*}
\item \textbf{Strong positivity}. For any $U_1,\ldots,U_n \in \Sigma_{\Lambda}$, the $n \times n$ matrix $\mathcal{D}(U_i, U_j)$ is positive semidefinite.
\end{itemize}
\end{definition}

As shown in \cite{dowker2010hilbert}, the decoherence functional allows for an alternative definition of the quantum measure: if $\mathcal{D}$ is a decoherence functional on a measurable space $\Lambda$ with $\sigma$-algebra $\Sigma_\Lambda$, then $\mu \colon \Sigma_\Lambda \to [0,1]$ given by $\mu(U) = \mathcal{D}(U,U)$ is a quantum measure over $\Sigma_\Lambda$.

Instead of the functor taking values in $\cat{SRel}$, one might envision a receiving category whose states are quantum measures rather than probability measures. Objects would still be measurable spaces, as before. But morphisms, instead of being Markov kernels, would now be functions $f \colon X \times \Sigma_Y \to [0,1]$ such that $f(-,V) \colon X \to [0,1]$ is a measurable function for each $V$, and $f(x,-) \colon \Sigma_Y \to [0,1]$ is a \emph{quantum} measure. This seems to be monoidal as before. 
The problem is that it is unclear how to define composition of such morphisms. We would like to say that $(g \circ f)(x,W)=\int g(y,W) f(x,\mathrm{d}y)$. But that needs a good notion of integration against quantum measures, as it is unclear whether $(g \circ f)(x,-)$ is again a well-defined quantum measure. Moreover, associativity of this composition seems to come down to a Fubini-type theorem. Such a theory of quantum integration seems only to be embryonic as of yet~\cite{gudder:integrals,gudder:integration}, presumably because so far quantum measures have mostly been used to model causal sets, in which context integrals do not naturally fit.

Trying to define the desired category as the Kleisli category of a Giry-like monad, that takes a measurable space $X$ to the set $Q(X)$ of \emph{quantum} measures on it, runs into similar issues. The unit, given by Dirac delta functions, is still well-defined because probability measures are certainly quantum measures. But it is unclear whether the natural candidate for the multiplication, that sends $\Phi \in Q(Q(X))$ to the function that assigns to $U \in \Sigma_X$ the number $\int \Phi(\{ \phi \in Q(X) \mid \phi(U)>t \})\mathrm{d}t$, is well-defined at all.

\end{document}